\documentclass[a4paper,USenglish,cleveref, autoref, thm-restate]{lipics-v2021}

\usepackage{xspace}

\newtheorem{question}[theorem]{Question}

\hideLIPIcs  

\newcommand{\defparproblem}[4]{
	\vspace{3mm}
	\noindent\fbox{
		\begin{minipage}{0.96\linewidth}
			\begin{tabular*}{\linewidth}{@{\extracolsep{\fill}}lr} \textsc{#1} & {\bf{Parameter:}} #3 \\ \end{tabular*}
			{\bf{Input:}} #2 \\
			{\bf{Question:}} #4
		\end{minipage}
	}
	\vspace{2mm}
}

\newcommand{\br}[1]{\left(#1\right)}

\newcommand{\Oh}{\mathcal{O}}
\newcommand{\poly}{\ensuremath{\text{poly}}}

\newcommand{\sub}{\subseteq}
\newcommand{\sm}{\setminus}
\newcommand{\gr}{\mathrm{Graver}}
\newcommand{\conp}{{NP} $\not\sub$ {coNP/poly}\xspace}
\newcommand{\conpunless}{{NP} $\sub$ {coNP/poly}\xspace}

\newcommand{\zz}{\ensuremath{\mathbb{Z}}}

\newcommand{\nn}{\ensuremath{\mathbb{N}}}

\newcommand{\fcal}{\ensuremath{\mathcal{F}}}

\newcommand{\sslog}{{\sc Subset Sum[$\log t$]}\xspace}
\newcommand{\ilp}{{\sc 0-1 ILP[$m$]}\xspace}
\newcommand{\monoilp}{{\sc Monotone 0-1 ILP[$m$]}\xspace}
\newcommand{\zeroilp}{{\sc 0-Sum 0-1 ILP[$m$]}\xspace}
\newcommand{\vass}{{\sc 0-1 Counter Machine[$\ell$]}\xspace}
\newcommand{\andsat}{{\sc AND-3SAT[$k$]}\xspace}

\newcommand{\bin}{\ensuremath{\mathsf{bin}}}
\newcommand{\intro}{\ensuremath{\mathsf{introduce}}}
\newcommand{\forget}{\ensuremath{\mathsf{forget}}}
\newcommand{\edge}{\ensuremath{\mathsf{edge}}}

\title{Does Subset Sum Admit Short Proofs?}
\author{Michał Włodarczyk}{University of Warsaw}{michal.wloda@gmail.com}{https://orcid.org/0000-0003-0968-8414}{Supported by Polish National Science Centre SONATA-19 grant number 2023/51/D/ST6/00155.}

\authorrunning{M. Włodarczyk} 

\Copyright{Michał Włodarczyk} 



\category{} 

\relatedversion{} 

\nolinenumbers 

\EventEditors{Juli\'{a}n Mestre and Anthony Wirth}
\EventNoEds{2}
\EventLongTitle{35th International Symposium on Algorithms and Computation (ISAAC 2024)}
\EventShortTitle{ISAAC 2024}
\EventAcronym{ISAAC}
\EventYear{2024}
\EventDate{December 8--11, 2024}
\EventLocation{Sydney, Australia}
\EventLogo{}
\SeriesVolume{322}
\ArticleNo{6}

\begin{document}

\maketitle

\begin{abstract}
    We investigate the question whether \textsc{Subset Sum} can be solved by a polynomial-time algorithm with access to a certificate of length $\mathrm{poly}(k)$ where $k$ is the maximal number of bits in an input number. 
    In other words, can it be solved using only few nondeterministic bits?

    
    This question has motivated us to initiate a systematic study of certification complexity of parameterized problems.  
    Apart from \textsc{Subset Sum}, we examine problems related to integer linear programming, scheduling, and group theory.
    We reveal an equivalence class of problems sharing the same hardness with respect to having a polynomial certificate.
    These include {\sc Subset Sum} and {\sc Boolean Linear Programming} parameterized by the number of constraints.
    Secondly, we present new techniques for establishing lower bounds in this regime.
    In particular, we show that {\sc Subset Sum} in permutation groups is at least as hard for nondeterministic computation as {\sc 3Coloring} in bounded-pathwidth graphs. 
\end{abstract}

\section{Introduction}

Nondeterminism constitutes a powerful lens for studying complexity theory.
The most prominent instantiation of this concept is the class NP 
capturing all problems with solutions checkable in polynomial time.
Another well-known example is the class NL of problems that can be solved nondeterministically in logarithmic space~\cite{arora2009computational}.
But the usefulness of nondeterminism is not limited to merely filtering candidates for deterministic classes.
A~question studied in proof complexity theory is how much nondeterminism is needed to solve certain problems or, equivalently, how long proofs have to be to prove certain theorems~\cite{cook2010logical}.
Depending on the considered logic, these theorems may correspond to instances of problems complete for NP~\cite{KoblerM00}, coNP~\cite{Cook_Reckhow_1979}, or W[SAT]~\cite{DantchevMS11}.
The central goal of proof complexity is to
establish lower bounds for increasingly powerful proof systems in the hope of building up techniques to prove, e.g., NP $\ne$ coNP. 
What is more, there are connections between nondeterministic running time lower bounds and fine-grained complexity~\cite{CarmosinoGIMPS16}.
In the context of online algorithms, nondeterminism is used to measure how much knowledge of future requests is needed to achieve a certain performance level~\cite{bockenhauer2017online, boyar2017online, dobrev2009measuring}.

Bounded nondeterminism plays an important role in organizing parameterized complexity theory.
The first class studied in this context was W[P] comprising parameterized problems solvable in FPT time, i.e., $f(k)\cdot\poly(n)$, when given access to $f(k)\cdot\log(n)$ nondeterministic bits~\cite{cygan2015parameterized}.
In the last decade, classes defined by nondeterministic computation in limited space have attracted significant attention~\cite{AllenderCLPT14, ElberfeldST15, PilipczukWrochna18} with a recent burst of activity around the class XNLP~\cite{BodlaenderGJJL22, BodlaenderGNS21, BodlaenderMOPL23} of problems solvable in nondeterministic time $f(k)\cdot\poly(n)$ and space $f(k)\cdot\log(n)$.
While the study of W[P] and XNLP concerns problems considered very hard from the perspective of FPT algorithms, a question that  has eluded a systematic examination so far 
is {\bf how much nondeterminism is necessary to solve FPT problems in polynomial time.} 
A related question has been asked about the amount of nondeterminism  needed to solve $d$-CNF-SAT in sub-exponential time~\cite{dantsin2011}.

To concretize our question,
we say that a parameterized problem $P$ admits a {\em polynomial certificate} if an instance $(I,k)$ can be solved in polynomial time when given access to $\poly(k)$ nondeterministic bits\footnote{It is more accurate to say that a ``certificate'' refers to a particular instance while a problem can admit a ``certification''. We have decided however to choose a shorter and more established term. We also speak of ``certificates'' instead of ``witnesses'' because ``witness'' sometimes refers to a concrete representation of a solution for problems in NP, see e.g.,~\cite{DruckerNS16}.}.
For example, every problem in NP admits a polynomial certificate under parameterization by the input length.
This definition captures, e.g., FPT problems solvable via branching as a certificate can provide a roadmap for the correct branching choices.
Furthermore, every parameterized problem that is in NP and admits a polynomial kernelization has a polynomial certificate given by the NP certificate for the compressed instance.
The containment in NP plays a subtle role here: Wahlstr{\"{o}}m~\cite{Wahlstrom13Abusing} noted that a~polynomial compression for the {\sc $K$-Cycle} problem is likely to require a target language from outside NP exactly because {\sc $K$-Cycle} does not seem to admit a polynomial certificate.

In this article we aim to
organize the folklore knowledge about polynomial certification into a systematic study, provide new connections, techniques, and motivations,
and lay the foundations for a hardness framework.

\subparagraph{When is certification easy or hard?}
The existence of a certificate of size $p(k) = \poly(k)$ entails an FPT algorithm with running time $2^{p(k)}\poly(n)$,  
by enumerating all possible certificates. We should thus restrict ourselves only to problems solvable within such running time.
On the other hand, when such an algorithm is available then one can solve the problem in polynomial time whenever $p(k) \le \log n$.
Therefore, it suffices to handle the instances with $\log n < p(k)$. 
Consequently, 
for such problems it is equivalent to ask for a certificate of size $\poly(k+\log n)$ as this can be bounded polynomially in $k$ via the mentioned trade-off (see \Cref{lem:prelim:log}).
This observation yields polynomial certificates for problems parameterized by the solution size, such as {\sc Multicut}~\cite{MarxR14} or {\sc Planarization}~\cite{JansenLS14}, which do not fall into the previously discussed categories.

What are the problems solvable in time $2^{k^{\Oh(1)}}n^{\Oh(1)}$ yet unlikely to admit a polynomial certificate?
The {\sc Bandwidth} problem has been conjectured 
not be belong to W[P] because one can merge multiple instances into one, without increasing the parameter, in such a way that the large instance is solvable if and only if all the smaller ones are.
It is conceivable that a certificate for the large instance should require at least one bit for each of the smaller instances, hence it cannot be short~\cite{FellowsR20}.
The same argument applies to every parameterized problem that admits an AND-composition, a construction employed to rule out polynomial kernelization~\cite{Dell16, fomin2019kernelization}, which effectively encodes a conjunction of multiple 3SAT instances as a single instance of the problem.
Such problems include those parameterized by graph width measures like treewidth or pathwidth, and it is hard to imagine polynomial certificates for them.
However, kernelization hardness can be also established using an OR-composition, which does not stand at odds with polynomial certification. 

The close connection between AND-composition and polynomial certificates has been observed
by Drucker, Nederlof, and Santhanam~\cite{DruckerNS16}
who focused on parameterized search problems solvable by
{\em One-sided Probabilistic Polynomial} (OPP) algorithms (cf. \cite{paturi2010}).
They asked which problems admit an OPP algorithm that finds a solution with probability $2^{-\poly(k)}$.
This may seem much more powerful than using $\poly(k)$ nondeterministic bits but the success probability can be replaced by $\Omega(1)$ when given a single access to an oracle solving $k$-variable Circuit-SAT~\cite[Lemma 3.5]{DruckerNS16}. 
A~former result of Drucker~\cite{Drucker13nondet} implies that an OPP algorithm with success probability $2^{-\poly(k)}$ (also called a {\em polynomial Levin witness compression}) for a search problem admitting a so-called constructive AND-composition would imply \conpunless~\cite[Theorem 3.2]{DruckerNS16}.




\subparagraph{Constructive vs. non-constructive proofs.}

The restriction to search problems is crucial in the work~\cite{DruckerNS16} because 
the aforementioned hardness result does not apply to algorithms that may
recognize yes-instances without constructing a solution explicitly but by proving its existence in a non-constructive fashion.
As noted by Drucker~\cite[\S 1.3]{Drucker13nondet-long}, his negative results do not allow to rule out this kind of algorithms.

In general, search problems may be significantly harder than their decision counterparts.
For example, there are classes of search problems for which the solution is always guaranteed to exist (e.g., by the pigeonhole principle, in the case of class PPP~\cite{aaronson2005complexity}) but the existence of a polynomial algorithm computing some solution is considered unlikely.
For a less obvious example, consider finding a
non-trivial divisor 
of a given integer $n$.
A polynomial (in $\log n$) algorithm finding a solution could be used to construct the factorization of $n$, resolving a major open problem.
But the existence of a solution is equivalent to $n$ being composite and this can be verified in polynomial time by the AKS primality test~\cite{agrawal2004primes}.

As yet another example, consider the problem of finding a knotless embedding of a graph, i.e., an embedding in $\mathbb{R}^3$ in which every cycle forms a trivial knot in a topological sense.
The class $\mathcal{G}$ of graphs admitting such an embedding is closed under taking minors so Robertson and Seymour's Theorem ensures that $\mathcal{G}$ is characterized by a finite set of forbidden minors~\cite{robertson2004graph}, leading to a polynomial algorithm for recognizing graphs from $\mathcal{G}$.
Observe that excluding all the forbidden minors yields a non-constructive proof that a knotless embedding exists.
On the other hand, the existence of a polynomial algorithm constructing such an embedding remains open~\cite{lazarus2017knots}.

To address this discrepancy, we propose the following conjecture which asserts that not only {\em finding} assignments to many instances of 3SAT requires many bits of advice but even {\em certifying} that such assignments exist should require many bits of advice.
We define the parameterized problem \andsat where an instance consists of a sequence of $n$ many 3SAT formulas on $k$ variables each, and an instance belongs to the language if all these formulas are satisfiable. We treat $k$ as a parameter.

\begin{conjecture}\label{conj:certificate}
     \andsat does not admit a polynomial certificate unless \conpunless.
\end{conjecture}

Observe that \andsat is solvable in time $2^k\cdot\poly(k)\cdot n$, so the questions whether it admits a certificate of size $\poly(k)$, $\poly(k)\cdot\log n$, or $\poly(k + \log n)$ are equivalent.
We formulate \Cref{conj:certificate} as a conditional statement because we believe that its proof in the current form is within the reach of the existing techniques employed in communication complexity and kernelization lower bounds~\cite{Dell16, paturi2010}.
Then the known examples of AND-composition could be interpreted as reductions from \andsat that justify non-existence of polynomial certificates.


\subsection{The problems under consideration}
\subparagraph{Our focus: Subset Sum.}

In {\sc Subset Sum} we are given a sequence of $n$ integers (also called items), a~target integer $t$ (all numbers encoded in binary), and we ask whether there is a subsequence summing up to~$t$.
This is a fundamental NP-hard problem that can be solved in pseudo-polynomial time $\Oh(tn)$ by the classic algorithm by Bellman from the 50s~\cite{bellman1958dynamic}.
In 2017 the running time has been improved to $\tilde\Oh(t+n)$ by Bringmann~\cite{Bringmann17}.
{\sc Subset Sum} reveals miscellaneous facets in complexity theory: it has been studied from the perspective of exponential algorithms~\cite{NederlofLZ12, NederlofW21}, logarithmic space~\cite{JinVW21, Kane10}, 
approximation~\cite{BringmannN21, chen2024approximating, KellererMPS03, MuchaW019}, kernelization~\cite{DomLS14, harnik2010compressibility, JansenR021}, fine-grained complexity~\cite{AbboudBHS22, BringmannW21, PolakRW21, pisinger1999linear}, 
cryptographic systems~\cite{ImpagliazzoN96}, and average-case analysis~\cite{merkle1978hiding}.
Our motivating goal is the following question.

\begin{question}\label{conj:ss}
    Does {\sc Subset Sum} admit a polynomial certificate for parameter $k = \log t$?
\end{question}

From this point of view, the pseudo-polynomial time $\Oh(tn)$ can be interpreted as FPT running time $\Oh(2^kn)$.
It is also known that a~kernelization of size $\poly(k)$ is unlikely~\cite{DomLS14}.
Observe that we cannot hope for a certificate of size $o(\log t)$ because the algorithm  enumerating all possible certificates would solve {\sc Subset Sum} in time $2^{o(\log t)}n^{\Oh(1)} = t^{o(1)}n^{\Oh(1)}$ contradicting the known lower bound based on the Exponential Time Hypothesis (ETH)~\cite{AbboudBHS22}.


The parameterization by the number of relevant bits exhibits a behavior different from those mentioned so far, that is, width parameters and solution size, making it an uncharted territory for nondeterministic algorithms.
 The study of \sslog  was suggested by Drucker et al.~\cite{DruckerNS16}  in the context of polynomial witness compression.
These two directions are closely related yet ultimately
incomparable: the requirement to return a solution makes the task more challenging but the probabilistic guarantee is less restrictive than constructing a certificate.
However, establishing hardness in both paradigms boils down to finding
a reduction of a certain kind from \textsc{AND-3SAT[$k$]}.

The {\em density} of a {\sc Subset Sum} instance is defined as $n \,/\, \log(t)$ in the cryptographic context~\cite{AustrinKKN16, howgrave2010new, lagarias1985solving}.
As it~is straightforward to construct a certificate of size $n$, we are mostly interested in instances of high density, which also appear hard for exponential algorithms~\cite{AustrinKKN16}.
On the other hand, instances that are very dense enjoy a special structure that can be leveraged algorithmically~\cite{BringmannW21, GalilM91}.
The instances that seem the hardest in our regime are those in which $n$ is slightly superpolynomial in $\log t$. 
Apart from the obvious motivation to better understand the structure of {\sc Subset Sum}, we believe that the existence of short certificates for dense instances could be valuable for cryptography. 

There are several other studied variants of the problem. In {\sc Unbounded Subset Sum} the input is specified in the same way but one is allowed to use each number repeatedly.
Interestingly, this modification enables us to certify a solution with 
 $\Oh(\log^2 t)$ bits (\Cref{lem:remain:unbounded}).
Another variant is to replace the addition with some group operation.
In {\sc Group-$G$ Subset Sum} we are given a sequence of $n$ elements from $G$ and we ask whether one can pick a subsequence whose group product  equals the target element $t \in G$.
Note that we do not allow to change the order of elements when computing the product what makes a difference for non-commutative groups. 
This setting has been mostly studied for $G$ being the cyclic group $\zz_q$~\cite{AxiotisBBJNTW21, AxiotisBJTW19, CardinalI21, KoiliarisX19, Potepa21}.
To capture the hardness of an instance, we choose the parameter to be $\log |G|$ (or equivalent).
In particular, we will see that {\sc Group-$\zz_{q}$ Subset Sum[$\log q$]} is equivalent in our regime to \sslog.
We will also provide examples of groups for which certification is either easy or conditionally hard. 

\subparagraph{Integer Linear Programming.}
We shall consider systems of equations in the form $\{Ax=b \mid x \in \{0,1\}^n\}$ where $A \in \zz^{m \times n}$, $b \in \zz^m$, with the parameter being the number of constraints $m$.
This is a special case of \textsc{Integer Linear Programming} (ILP) over boolean domain, known in the literature as  {\em pseudo-boolean optimization}.
This case has been recognized as particularly interesting by the community working on practical ILP solvers because pseudo-boolean optimization can be treated with SAT solvers \cite{aloul2002generic, BussN21, een2006translating, JoshiMM15, SmirnovBJ21}.
It is also applicable in the fields of approximation algorithms \cite{shmoys} and election systems \cite{lin2012solving}.
Eisenbrand and Weismantel~\cite{EisenbrandW20} found an elegant application of Steinitz Lemma to this problem and gave  
an FPT algorithm
with running time $(||A||_\infty + m)^{\Oh(m^2)}\cdot n$ (cf. \cite{JansenR19}).

For simplicity we will consider variants with bounded $||A||_\infty$.
In the {\sc 0-1 ILP} problem we restrict ourselves to matrices $A \in \{-1,0,1\}^{m \times n}$ and in  {\sc Monotone 0-1 ILP} we consider $A \in \{0,1\}^{m \times n}$.
A~potential way to construct a short certificate would be to tighten the {\em proximity bounds} from~\cite{EisenbrandW20} for such matrices: find an extremal solution $x^*$ to the linear relaxation  $\{Ax=b \mid x \in [0,1]^n\}$ and hope that some integral solution $z$ lies nearby, i.e., $||z-x^*||_1 \le \poly(m)$.
This however would require tightening the bounds on the vector norms in the  {\em Graver basis} of the matrix $A$.
Unfortunately, there are known lower bounds making this approach hopeless~\cite{berndt2024new, berstein2009graver, kudo2013lower}.

It may be tempting to seek the source of hardness in the large values in the target vector~$b$.
We will see however that the problem is no easier when we assume $b=0$ and look for any non-zero solution. 
We refer to such problem as {\sc 0-Sum~0-1~ILP}.

A special case of {\sc Monotone 0-1 ILP} is given by a matrix $A \in \{0,1\}^{m \times n}$ with $n = {\binom{m}{d}}$ columns corresponding to all size-$d$ subsets of $\{1,\dots,m\}$.
Then $Ax = b$ has a boolean solution if and only if $b$ forms a degree sequence of some $d$-hypergraph, i.e., it is $d$-hypergraphic.
There is a classic criterion by Erd\H{o}s for a sequence to be graphic~\cite{erdos1960graphs} (i.e., 2-hypergraphic) but already for $d=3$ deciding if $b$ is $d$-hypergraphic becomes NP-hard~\cite{DezaLMO19}.
It is straightforward to certify a solution with $m^d$ bits, what places the problem in NP for each fixed $d$, but it is open whether it is in NP when $d$ is a part of the input (note that the matrix $A$ is implicit so the input size is $\Oh(md\log m)$).
This basically boils down to the same dilemma: {\bf can we certify the existence of a boolean ILP solution $x$  without listing $x$ in its entirety?}

There is yet another motivation to study the certification complexity of {\sc 0-1 ILP}.
If this problem admits a certificate of size $\poly(m)$ then any other problem that can be modeled by {\sc 0-1 ILP} with few constraints must admit a short certificate as well.
This may help classifying problems into these that can be solved efficiently with ILP solvers and those that cannot.

\subparagraph{Other problems parameterized by the number of relevant bits.}

A classic generalization of {\sc Subset Sum} is the {\sc Knapsack} problem where each item is described by a size $p_i$ and a weight~$w_i$; here we ask for a subset of items of total weight at least $w$ but total size not exceeding $t$. 
Following Drucker et al.~\cite{DruckerNS16} we parameterize it by the number of bits necessary to store items' sizes and weights, i.e., $\log(t + w)$.
The need to process weights makes {\sc Knapsack} harder than {\sc Subset Sum} from the perspective of fine-grained complexity~\cite{CyganMWW19} but they are essentially equivalent on the ground of exponential algorithms~\cite{NederlofLZ12}.
We will see that they are equivalent in our regime as well. 

We will also consider the following scheduling problem which is in turn a generalization  of {\sc Knapsack}.
In {\sc Scheduling Weighted Tardy Jobs}
we are given a set of $n$ jobs, where each
job $j \in [n]$ has a processing time $p_j \in\nn$, a weight $w_j \in \nn$, and a due date $d_j \in \nn$. 
We schedule the jobs on a single machine and we want to minimize the total weight of jobs completed after their due dates (those jobs are called {\em tardy}).
Equivalently, we try to maximize the total weight of jobs completed in time.

In the scheduling literature, {\sc Scheduling Weighted Tardy Jobs} is referred to as $1 || \sum w_jU_j$ using Graham’s notation.
The problem is solvable in pseudo-polynomial time  $\Oh(n\cdot d_{max})$ by the classic Lawler and Moore's algorithm~\cite{lawler1969functional}.
The interest in $1 || \sum w_jU_j$ has been revived due to the recent advances in fine-grained complexity
\cite{AbboudBHS20, BringmannFHSW22, fischer2024minimizing, Klein0R23}. 
Here, the parameter that captures the number of relevant bits is 
$\log (d_{max} + w_{max})$.

\subsection{Our contribution}

The standard {\em polynomial parameter transformation} (PPT) is a polynomial-time reduction between parameterized problems that maps an instance with parameter $k$ to one with parameter $k' = \poly(k)$.
We introduce the notion of a {\em nondeterministic polynomial parameter transformation} (NPPT) which extends PPT by allowing the reduction to guess ${\poly(k)}$ nondeterministic bits.
Such reductions preserve the existence of a polynomial certificate.
We write $P \le_{\textsc{nppt}} Q$ (resp. $P \le_{\textsc{ppt}} Q$) to indicate that $P$ admits a NPPT (resp. PPT) into $Q$.

We demonstrate how NPPT help us organize the theory of polynomial certification, similarly as  PPT come in useful for organizing the theory of Turing kernelization~\cite{HermelinKSWW15}. 
As our first result, we present an equivalence class of problems that share the same certification-hardness status as \sslog.
In other words, either all of them admit a polynomial certificate or none of them.
Despite apparent similarities between these problems, some of the reductions require a nontrivial use of nondeterminism.

\begin{theorem}\label{thm:main:equivalent}
    The following parameterized problems are equivalent with respect to NPPT:
    \begin{enumerate}
        \item \sslog
        \item {\sc Knapsack[$\log(t+w)$]}, {\sc Knapsack[$\log(p_{max}+w_{max})$]}
        \item \ilp, \monoilp, {\sc 0-Sum 0-1 ILP[$m$]}
        \item {\sc Group-$\zz_q$ Subset Sum[$\log q$]}
    \end{enumerate}
\end{theorem}

Even though we are unable to resolve \Cref{conj:ss}, we believe that revealing such an equivalence class supports the claim that a polynomial certificate for \sslog is unlikely.
Otherwise, there must be some intriguing common property of all problems listed in \Cref{thm:main:equivalent} that has eluded researchers so far despite extensive studies in various regimes. 

Next, we present two negative results.
They constitute a proof of concept that 
\andsat can be used as a non-trivial source of hardness.
First, we adapt a reduction from~\cite{AbboudBHS20} to show that scheduling with weights and due dates is hard assuming \Cref{conj:certificate}.

\begin{restatable}[$\bigstar$]{theorem}{thmMainScheduling}
\label{thm:main:scheduling}
    \andsat $\le_{\textsc{ppt}}$ \mbox{{\sc Scheduling Weighted Tardy Jobs[$\log (d_{max} + w_{max})$]}}.
\end{restatable}

It is possible to formulate this result in terms of AND-composition but we chose not to work with this framework since it is tailored for refuting kernelization and relies on concepts that do not fit into our regime (e.g., polynomial relation~\cite{fomin2019kernelization}).

Our second hardness result involves {\sc Group-$S_k$ Subset Sum[$k$]}: a variant of {\sc Subset Sum} on permutation groups.
Such groups contain exponentially-large cyclic subgroups (see \Cref{lem:prelim:landau}) so this problem is at least as hard as {\sc Group-$\zz_q$ Subset Sum[$\log q$]} (which is equivalent to \sslog).
We reduce from {\sc 3Coloring} parameterized by pathwidth which is at least as hard as \andsat with respect to PPT.
Indeed, we can transform each 3SAT formula in the input (each of size $\Oh(k^3)$) into an instance of {\sc 3Coloring} via the standard NP-hardness proof, and take the disjoint union of such instances, which implies \andsat $\le_{\textsc{ppt}}$ {\sc 3Coloring[pw]}.
Notably, the reduction in the other direction is unlikely (see \Cref{lem:remain:3color-mk2}) so 
{\sc 3Coloring[pw]} is probably harder than \andsat.

\begin{restatable}{theorem}{thmMainPermutation}
\label{thm:main:permutation} 
    {\sc 3Coloring[pw]} $\le_{\textsc{nppt}}$ {\sc Group-$S_k$ Subset Sum[$k$]}.
\end{restatable}

Consequently, {\sc Group-$S_k$ Subset Sum[$k$]}
does not admit a polynomial certificate assuming \Cref{conj:certificate} and NP $\not\sub$ coNP/poly.
Unlike \Cref{thm:main:scheduling}, this time establishing hardness requires a nondeterministic reduction.
An interesting feature of {\sc 3Coloring[pw]} is that it is NL-complete under logspace reductions when
the pathwidth {\sc pw} is restricted to $\Oh(\log n)$~\cite{AllenderCLPT14, PilipczukWrochna18}. 
On the other hand, {\sc Subset Sum} can be solved in time $\tilde{\Oh}(tn^2)$ and space $\text{polylog}(tn)$ using algebraic techniques~\cite{JinVW21}.
Therefore, obtaining a logspace PPT from {\sc 3Coloring[pw]} to \sslog (where \textsc{pw} = $\log n$ implies $t = 2^{\text{polylog}(n)}$) 
would lead to a surprising consequence: a proof that 
NL $\sub$ DSPACE($\text{polylog}(n)$) that is significantly different from Savitch's Theorem (see also discussion in~\cite[\S 1]{PilipczukWrochna18} on low-space determinization).
This suggests that a hypothetical reduction to \sslog should either exploit the ``full power'' of NPPT (so it cannot be improved to a logspace PPT) or 
start directly from \andsat.

 Finally, we examine the case of the group family $\zz^k_k$ on which {\sc Subset Sum} is still NP-hard (as this generalizes {\sc Subset Sum} on cyclic groups) 
 but enjoys a polynomial certificate.
Specifically, 
we exploit the bound on the maximal order of an element in $\zz^k_k$ to prove that
 there always exists a solution of bounded size.

\begin{restatable}[$\bigstar$]{lemma}{thmZkk}
    {\sc Group-$\zz^k_k$ Subset Sum[$k$]} admits a polynomial certificate.
\end{restatable}

In summary, {\sc Group-$G$ Subset Sum} appears easy for $G=\zz^k_k$ (due to bounded maximal order), hard for $G = S_k$ (due to non-commutativity), and the case $G = \zz_{2^k}$ lies somewhere in between. 
In the light of \Cref{thm:main:equivalent}, tightening this gap seems a promising avenue to settle \Cref{conj:ss}.

\subparagraph{Organization of the paper.}
We begin with the preliminaries where we formally introduce the novel concepts, such as NPPT.
We prove Theorems~\ref{thm:main:equivalent} and \ref{thm:main:permutation} in Sections \ref{sec:subsetsum} and \ref{sec:perm}, respectively.
The proofs marked with ($\bigstar$) can be found in the appendix.

\section{Preliminaries}
\label{sec:prelims}

We denote the set $\{1,\dots,n\}$ by $[n]$.
For a sequence $x_1, x_2, \dots, x_n$, its subsequence is any sequence of the form $x_{i_1}, \dots, x_{i_m}$ for some choice of increasing indices $1 \le i_1 < \dots <i_m \le n$.
All considered logarithms are binary.

A parameterized problem $P$ is formally defined as a subset of $\Sigma^* \times \nn$.
For the sake of disambiguation, whenever we refer to a parameterized problem, we denote the choice of the parameter in the $[\cdot]$ bracket, e.g., {\sc 3Coloring[pw]}.
We call $P$ {\em fixed-parameter tractable} (FPT) is the containment $(I,k) \in P$ can be decided in time $f(k)\cdot\poly(|I|)$ for some computable function $f$.
We say that $P$ admits a {\em polynomial compression} into a problem $Q$ if there is a polynomial-time algorithm that transform $(I,k)$ into an equivalent instance of $Q$ of size $\poly(k)$.
If $Q$ coincides with the non-parameterized version of $P$ then such an algorithm is called a {\em polynomial kernelization}.
A {\em polynomial Turing kernelization} for $P$ is a polynomial-time algorithm that determines if $(I,k) \in P$ using an oracle that can answer if $(I',k') \in P$ whenever $|I'|+k' \le \poly(k)$.

\begin{definition}
    Let $P \sub \Sigma^* \times \nn$ be a parameterized problem.
    We say that $P$ has a \underline{polynomial certificate} if there is an algorithm ${\mathcal A}$ that, given an instance $(I,k)$ of $P$ and a string $y$ of $\poly(k)$ bits, runs in polynomial time and accepts or rejects $(I,k)$ with the following guarantees.
    \begin{enumerate}
        \item If $(I,k) \in P$, then there exists $y$ for which ${\mathcal A}$ accepts.
        \item If $(I,k) \not\in P$, then ${\mathcal A}$ rejects  $(I,k)$ for every $y$.
    \end{enumerate}
\end{definition}

\begin{lemma}
    Let $P \sub \Sigma^* \times \nn$ and $Q \sub \Sigma^*$.
    Suppose that $Q \in NP$ and $P$ admits a polynomial compression into $Q$.
    Then $P$ admits a polynomial certificate.
\end{lemma}
\begin{proof}
    For a given instance $(I,k)$ of $P$ we execute the compression algorithm to obtain an equivalent instance $I'$ of $Q$ of size $\poly(k)$.
    Since  $Q \in$ NP the instance $I'$ can be solved in polynomial-time with an access to a string $y$ of $\poly(|I'|) = \poly(k)$ nondeterministic bits.
    Then $y$ forms a certificate for $(I,k)$.
\end{proof}

\begin{definition}
    Let $P,Q \sub \Sigma^* \times \nn$ be parameterized problems.
    An algorithm ${\mathcal A}$ is called a \underline{polynomial parameter transformation} (PPT) from $P$ to $Q$ if, given an instance $(I,k)$ of $P$, runs in polynomial time, and outputs an equivalent instance $(I',k')$ of $Q$ with $k' \le \poly(k)$.

    An algorithm ${\mathcal B}$ is called a \underline{nondeterministic polynomial parameter transformation} (NPPT) from $P$ to $Q$ if, given an instance $(I,k)$ of $P$ and a string $y$ of $\poly(k)$ bits, runs in polynomial time, and outputs an instance $(I',k')$ of $Q$ with the following guarantees.
    \begin{enumerate}
        \item $k' \le \poly(k)$
        \item If $(I,k) \in P$, then there exists $y$ for which ${\mathcal B}$ outputs $(I',k') \in Q$.
        \item If $(I,k) \not\in P$, then ${\mathcal B}$ outputs $(I',k') \not\in Q$ for every $y$.
    \end{enumerate}
\end{definition}

Clearly, PPT is a special case of NPPT.
We write $P \le_{\textsc{ppt}} Q$ ($P \le_{\textsc{nppt}} Q$) if there is a (nondeterministic) PPT from $P$ to $Q$.
We write $P \equiv_{\textsc{ppt}} Q$ ($P \equiv_{\textsc{nppt}} Q$) when we have reductions in both directions.
It is easy to see that the relation $\le_{\textsc{nppt}}$ is transitive.
Similarly as the relation $\le_{\textsc{ppt}}$ is monotone with respect to having a polynomial kernelization, the relation $\le_{\textsc{nppt}}$
is monotone with respect to having a polynomial certificate. 

\begin{lemma}
     Let $P,Q \in \Sigma^* \times \nn$ be parameterized problems. If  $P \le_{\textsc{nppt}} Q$ and $Q$ admits a polynomial certificate then $P$ does as well.
\end{lemma}
\begin{proof}
    Given an instance $(I,k)$ of $P$ the algorithm guesses a string $y_1$ of length $\poly(k)$ guiding the reduction to $Q$ and constructs an instance $(I',k')$ with $k' = \poly(k)$.
    Then it tries to prove that $(I',k') \in Q$ by guessing a certificate $y_2$ of length  $\poly(k') = \poly(k)$.
\end{proof}

A different property transferred by PPT is {polynomial Turing kernelization.}
Hermelin et al.~\cite{HermelinKSWW15} proposed a hardness framework for this property by considering complexity classes closed under PPT (the WK-hierarchy).

Next, we prove the equivalence mentioned in the Introduction.

\begin{lemma}\label{lem:prelim:log}
    Suppose $P \sub \Sigma^* \times \nn$ admits an algorithm $\mathcal{A}$ deciding if $(I,k) \in P$ in time $2^{p(k)}\poly(|I|)$  where $p(\cdot)$ is a polynomial function. 
    Then $P[k] \equiv_{\textsc{ppt}} P[k+\log |I|]$.
\end{lemma}
\begin{proof}
    The direction  $P[k+\log |I|] \le_{\textsc{ppt}} 
 P[k]$ is trivial.
    To give a reduction in the second direction, we first check if $p(k) \le \log |I|$.
    If yes, we execute $\mathcal{A}$ in time $\poly(|I|)$ and according to the outcome we return a trivial yes/no-instance.
    Otherwise we have $\log |I| < p(k)$ so we can output $(I,k')$ for the new parameter $k' = k + \log |I|$ being polynomial in $k$.
\end{proof}

\subparagraph{Pathwidth.}
A path decomposition of a graph $G$ is a
sequence $\mathcal{P} = (X_1,X_2,\dots,X_r)$ of {\em bags}, where $X_i \sub V(G)$, and:
\begin{enumerate}
    \item For each~$v \in V(G)$ the set~$\{i \mid v \in X_i\}$ forms a {non-empty} subinterval of $[r]$. 
    \item For each edge~$uv \in E(G)$ there is $i \in [r]$ with~$\{u,v\} \subseteq X_i$.
\end{enumerate}

The \emph{width} of a path decomposition is defined as~$\max_{i=1}^r |X_i| - 1$. The \emph{pathwidth} of a graph~$G$ is the minimum width of a path decomposition of~$G$.

\begin{lemma}{\cite[Lemma 7.2]{cygan2015parameterized}}
\label{prelim:pathwidth:nice}
    If a graph $G$ has pathwidth at most $p$, then it admits a {\em nice} path decomposition $\mathcal{P} = (X_1,X_2,\dots,X_r)$ of width at most $p$, for which:
\begin{itemize}
    \item $X_1 = X_r = \emptyset$.
    \item For each $i \in [r-1]$ there is either a vertex $v \not\in X_i$ for which $X_{i+1} = X_i \cup \{v\}$ or a vertex $v \in X_i$ for which $X_{i+1} = X_i \sm \{v\}$.
\end{itemize}
Furthermore, given any path decomposition of $G$, we can turn it into a nice path decomposition of no greater width, in polynomial time.
\end{lemma}

The bags of the form $X_{i+1} = X_i \cup \{v\}$ are called {\em introduce bags} while the ones of the form $X_{i+1} = X_i \sm \{v\}$ are called {\em forget bags}.

Similarly as in the previous works~\cite{AllenderCLPT14, BodlaenderGNS21, PilipczukWrochna18} we assume that a path decomposition of certain width is provided with the input.
This is not a restrictive assumption for our model since pathwidth can be approximated within a polynomial factor in polynomial time~\cite{GroenlandJNW23}.


\subparagraph{Group theory.}
The basic definitions about groups can be found in the book~\cite{robinson2003introduction}.
A~homomorphism between groups $G,H$ is a mapping $\phi\colon G \to H$  that preserves the group operation, i.e., $\phi(x) \circ_H \phi(y) = \phi(x \circ_G y)$ for all $x,y \in G$.
An isomorphism is a bijective homomorphism and
an automorphism of $G$ is an isomorphism from $G$ to $G$. 
We denote by $Aut(G)$ the automorphism group of $G$ with the group operation given as functional composition.
A~subgroup $N$ of $G$ is {\em normal} if for every $g \in G, n \in N$ we have $g \circ_G n \circ_G g^{-1} \in N$.

The symmetric group $S_k$ comprises permutations over the set $[k]$ with the group operation given by composition.
For a permutation $\pi \in S_k$ we consider a directed graph over the vertex set $[k]$ and arcs given as $\{(v,\pi(v)) \mid v \in [k]\}$.
The cycles of this graph are called the cycles~of~$\pi$.

We denote by $\zz_k$ the cyclic group with addition modulo $k$.
We write the corresponding group operation as $\oplus_k$.
An order of an element $x \in G$ is the size of the cyclic subgroup of $G$ generated by $x$.
The Landau's function $g(k)$ is defined as the maximum order of an element $x$ in~$S_k$. It is known that $g(k)$ equals $\max \mathsf{lcm}(k_1,\dots,k_\ell)$ over all partitions $k = k_1 + \dots + k_\ell$ (these numbers correspond to the lengths of cycles in $x$) and that $g(k) = 2^{\Theta(\sqrt{k\log k})}$~\cite{Nicolas1997}. 
An element of large order can be found easily if we settle for a slightly weaker bound.

\begin{lemma}\label{lem:prelim:landau}
    For each $k$ there exists $\pi \in S_k$ of order $2^{\Omega({\sqrt k}/{\log k})}$ and it can be found in time $\poly(k)$.
\end{lemma}
\begin{proof}
    Consider all the primes $p_1,\dots,p_\ell$ that are smaller than $\sqrt{k}$.
    By the prime number theorem there are $\ell = \Theta({\sqrt k}/{\log k})$ such primes~\cite{hardy1979introduction}.
    We have $p_1 + \dots + p_\ell \le {\sqrt k} \cdot {\sqrt k} = k$ so we can find a permutation in $S_k$ with cycles of lengths $p_1,\dots,p_\ell$ (and possibly trivial cycles of length 1).
    We have $\mathsf{lcm}(p_1,\dots,p_\ell) = \prod_{i=1}^\ell p_i \ge 2^\ell = 2^{\Omega({\sqrt k}/{\log k})}$.
\end{proof}

For two groups $N, H$ and a homomorphism $\phi \colon H \to Aut(N)$ we define the {\em outer semidirect product}~\cite{robinson2003introduction} $N \rtimes_\phi H$ as follows.
The elements of $N \rtimes_\phi H$ are $\{(n,h) \mid n \in N, h \in H\}$ and the group operation $\circ$ is given as $(n_1,h_1) \circ (n_2,h_2) = \br{n_1 \circ \phi_{h_1}(n_2), h_1 \circ h_2}$.
A special case of the semidirect product occurs when we combine subgroups of a common group.

\begin{lemma}[{\cite[{\S 4.3}]{robinson2003introduction}}]
\label{lem:prelim:semidirect}
Let $G$ be a group with a normal subgroup $N$ and a subgroup $H$, such that every element $g \in G$ can be written uniquely 
as $g = n\circ h$ for $n \in N, h \in H$.
Let $\phi \colon H \to Aut(N)$ be given as $\phi_h(n) = h\circ n \circ h^{-1}$ (this is well-defined because $N$ is normal in $G$).
Then $G$ is isomorphic to the semidirect product $N \rtimes_\phi H$.
\end{lemma}

\defparproblem{Group-$G$ Subset Sum}{A sequence of elements $g_1, g_2,\dots,g_n \in G$, an element $g \in G$}{$\log |G|$}{Is there a subsequence $(i_1 < i_2 < \dots < i_r)$ of  $[n]$
such that $g_{i_1} \circ g_{i_2} \circ \dots \circ g_{i_r} = g$?}

We assume that the encoding of the group elements as well as the group operation $\circ$ are implicit for a specific choice of a group family.
For a group family parameterized by $k$, like $(S_k)_{k=1}^{\infty}$, we treat $k$ as the parameter.
In all considered cases it holds that $k \le \log |G| \le \poly(k)$ so these two parameterizations are equivalent under PPT.

\section{Equivalences}
\label{sec:subsetsum}

We formally introduce the variants of ILP that will be studied in this section.

\defparproblem{0-1 ILP}
{A matrix $A \in \{-1,0,1\}^{m \times n}$, a vector $b \in \zz^m$}
{$m$}
{Is there a vector $x \in \{0,1\}^n$ for which $Ax=b$?}

In {\sc Monotone 0-1 ILP} we restrict ourselves to matrices $A \in \{0,1\}^{m \times n}$.
In {\sc 0-Sum 0-1 ILP} we have $A \in \{-1,0,1\}^{m \times n}$ and we seek a binary vector $x \ne 0$ for which $Ax=0$.

{\bf We first check that all the parameterized problems considered in this section are solvable in time $2^{k^{\Oh(1)}}n^{\Oh(1)}$.}
For \sslog we can use the classic $\Oh(tn)$-time algorithm~\cite{bellman1958dynamic} which can be easily modified to solve {\sc Group-$\zz_q$ Subset Sum[$\log q$]} in time $\Oh(qn)$.
For {\sc Knapsack[$\log(p_{max} + w_{max})$]} there is an $\Oh(p_{max}\cdot w_{max}\cdot n)$-time algorithm~\cite{pisinger1999linear} which also works for the larger parameterization by $\log(t+w)$.
Next, \ilp can be solved in time $2^{\Oh(m^2\log m)}\cdot n$ using the algorithm for general matrix $A$~\cite{EisenbrandW20}.
This algorithm can be used to solve {\sc 0-sum 0-1 ILP[$m$]} due to \Cref{lem:ss:0sum:obs}.
Hence by \Cref{lem:prelim:log} we can assume in our reductions that $(\log n)$ is bounded by a 
polynomial function of the parameter.

\begin{lemma}
    {\sc Knapsack[$\log(t + w)$]} $\equiv_{\textsc{ppt}}$ {\sc Knapsack[$\log(p_{max} + w_{max})$]}.
\end{lemma}
\begin{proof}
    We only need to show the reduction from {\sc Knapsack[$\log(p_{max} + w_{max})$]}.
    When $p_{max} \cdot n < t$ we can afford taking all the items. 
    On the other hand, if $w_{max} \cdot n < w$ then no solution can exist.
    Therefore, we can assume that  $\log t \le \log p_{max} + \log n$ and $\log w \le \log w_{max} + \log n$.
    By \Cref{lem:prelim:log} we can assume $\log n$ to be polynomial in $\log(p_{max} + w_{max})$ so the new parameter $\log(t + w)$ is polynomial in the original one.   
\end{proof}

\begin{lemma}
    \sslog $\equiv_{\textsc{nppt}}$ {\sc Knapsack[$\log(t + w)$]}.
\end{lemma}
\begin{proof}
    The $(\le)$ reduction is standard: we translate each input integer $p_i$ into an item $(p_i,p_i)$ and set $w=t$.
    Then we can pack items of total weight $t$ into a knapsack of capacity $t$ if and only if the {\sc Subset Sum} instance is solvable.

    Now consider the $(\ge)$ reduction.
    Let $k = \log(t + w)$.
    By the discussion at the beginning of this section we can assume that $\log n \le \log(t\cdot w) \le 2k$.
    We can also assume that $w_{max} < w$ as any item with weight exceeding $w$ and size fitting into the knapsack would form a trivial solution.
    Let $W = w \cdot n + 1$.

    Suppose there is a set of items with total size equal $t' \le t$ and total weight equal $w' \ge w$.
    Note that $w'$ must be less than $W$.
    We nondeterministically choose $t'$ and $w'$: this requires guessing $\log t + \log W \le 4k$ bits.
    Now we create an instance of {\sc Subset Sum} by mapping each item $(p_i, w_i)$ into integer $p_i \cdot W + w_i$ and setting the target integer to $t'' = t' \cdot W + w'$.
    If we guessed $(t',w')$ correctly then such an instance clearly has a solution.
    On the other hand, if this instance of {\sc Subset Sum} admits a solution then we have $\sum_{i \in I} (p_i \cdot W + w_i) =  t' \cdot W + w'$ for some  $I \sub [n]$. 
    Since both $w'$ and $\sum_{i \in I'} w_i$ belong to $[1,W)$ we must have $\sum_{i \in I} w_i = w'$ and $\sum_{i \in I} p_i = t'$ so the original instance of {\sc Knapsack} has a solution as well.
    Finally, it holds that $\log t'' \le 5k$ so the parameter is being transformed linearly. 
\end{proof}

We will need the following extension of the last argument.

\begin{lemma}\label{lem:ss:twoseq}
    Let $W \in \nn$ and $a_1,\dots,a_n$, $b_1,\dots,b_n$ be sequences satisfying $a_i,b_i \in [0,W)$ for each $i\in[n]$.
    Suppose that $S := \sum_{i=1}^n a_i W^{i-1} = \sum_{i=1}^n b_i W^{i-1}$.
    Then $a_i=b_i$ for each $i\in [n]$.
\end{lemma}
\begin{proof}
    Consider the remainder of $S$ when divided by $W$.
    Since $W$ divides all the terms in $S$ for $i \in [2,n]$ and $a_1,b_1 \in [0,W)$ we must have $a_1 = (S \mod W) = b_1$.
    Next, consider $S' = (S - a_1) / W = \sum_{i=2}^n a_i W^{i-2} = \sum_{i=2}^n b_i W^{i-2}$.
    Then $a_2 = (S' \mod W) = b_2$. This argument generalizes readily to every $i \in [n]$.
\end{proof}

\begin{lemma}
    \sslog $\equiv_{\textsc{nppt}}$ \monoilp.
\end{lemma}
\begin{proof}
    ($\le$): Consider an instance $ (\{p_1, \dots, p_n\}, t)$ of {\sc Subset Sum}. Let $k = \lceil \log t \rceil$. We can assume that all numbers $p_i$ belong to the interval $[t]$.
    For an integer $x \in [t]$ let $\bin(x) \in \{0,1\}^k$ denote the binary encoding of $x$ so that $x = \sum_{j=1}^k \bin(x)_j \cdot 2^{j-1}$.
    Observe that the condition $x_1 + \dots + x_m = t$ can be expressed as $\sum_{j=1}^k \br{\sum_{i=1}^m \bin(x_i)_j} \cdot 2^{j-1} = t$.
    
    We nondeterministically guess a sequence $b = (b_1,\dots,b_k)$
    so that $b_j$ equals $\sum_{i \in I} \bin(p_i)_j$ where $I \sub [n]$ is a solution.
    This sequence must satisfy $\max_{j=1}^k b_j \le t$, and so we need $k^2$ nondeterministic bits to guess $b$.
    We check if the sequence $b$ satisfies $\sum_{j=1}^k b_j \cdot 2^{j-1} = t$; if no then the guess was incorrect and we return a trivial no-instance.
    Otherwise we construct an instance of {\sc Monotone 0-1 ILP[$k$]} with a system $Ax = b$.
    The vector $b$ is given as above and its length is $k$.
    The matrix $A$ comprises $n$ columns where the $i$-th column is $\bin(p_i)$.
    This system has a solution $x \in \{0,1\}^n$ if and only if there exists  $I \sub [n]$ so that $\sum_{i \in I} \bin(p_i)_j = b_j$ for all $j \in [k]$.
    This implies that $\sum_{i \in I} p_i = t$.
    Conversely, if such a set $I \sub [n]$ exists, then there is $b \in [t]^k$
    for which $Ax=b$ admits a boolean solution.

    ($\ge$): Consider an instance $Ax=b$ of \monoilp.  As usual, we assume $\log n \le \poly(m)$.
    We can also assume that $||b||_\infty \le n$ as otherwise $Ax=b$ is clearly infeasible.
    We construct an instance of {\sc Subset Sum} with $n$ items and target integer $t = \sum_{j=1}^m b_j \cdot (n+1)^{j-1}$.
    Note that $t \le m\cdot ||b||_\infty \cdot (n+1)^m$ so $\log t \le \poly(m)$.
    For $i \in [n]$ let $a^i \in \{0,1\}^m$ denote the $i$-th column of the matrix $A$.
    We define $p_i = \sum_{j=1}^m a^i_j \cdot (n+1)^{j-1}$ and we claim that that instance $J = (\{p_1, \dots, p_n\}, t)$ of {\sc Subset Sum} is solavble exactly when the system $Ax=b$ has a boolean solution.

    First, if $x \in \{0,1\}^m$ forms a solution to $Ax = b$ then for each $j\in [m]$ we have $\sum_{i=1}^n x_ia^i_j (n+1)^{j-1} = b_j(n+1)^{j-1}$ and so $\sum_{i=1}^n x_ip_i = t$.
    Hence  the set $I = \{i \in [n] \mid x_i = 1\}$ encodes a solution to $J$.
    In the other direction, suppose that there is $I \sub [n]$ for which $\sum_{i \in I} p_i = t$.
    Then $t =  \sum_{j=1}^m ( \sum_{i \in I} a^i_j) \cdot (n+1)^{j-1}$.
    Due to  \Cref{lem:ss:twoseq} we must have $b_i = \sum_{i \in I} a^i_j$ for each $i\in[m]$ and there is subset of columns of $A$ that sums up to the vector $b$.
    This concludes the proof.
\end{proof}

For the next reduction, we will utilize the lower bound on the norm of vectors in a so-called Graver basis of a matrix.
For two vectors $y,x \in \zz^n$ we write $y \vartriangleleft x$ if for every $i \in [n]$ it holds that $y_ix_i \ge 0$ and $|y_i| \le |x_i|$.
A non-zero vector $x \in \zz^n$ belongs to the Graver basis of $A \in \zz^{m \times n}$ if $Ax = 0$ and no other non-zero solution $Ay=0$ satisfies $y \vartriangleleft x$.
In other words, $x$ encodes a sequence of columns of $A$, some possibly repeated or negated, that sums to 0 and none of its nontrivial subsequences sums to 0.
The following lemma concerns the existence of vectors with a large $\ell_1$-norm in a Graver basis of a certain matrix.
We state it in the matrix-column interpretation.

\begin{lemma}[{\cite[Thm. 9, Cor. 5]{berndt2024new}}]\label{lem:graver:contruction}
    For every $k \in \nn$ there is a sequence $(v_1,\dots,v_n)$ of vectors from $\{-1,0,1\}^k$ such that
    \begin{enumerate}
        \item $n = \Theta(2^k)$,
        \item the vectors $v_1,\dots,v_n$ sum up to 0, and
        \item no proper non-empty subsequence of $(v_1,\dots,v_n)$ sums up to 0.
    \end{enumerate}
\end{lemma}

\begin{figure}[t]
\centering
\includegraphics[scale=0.6]{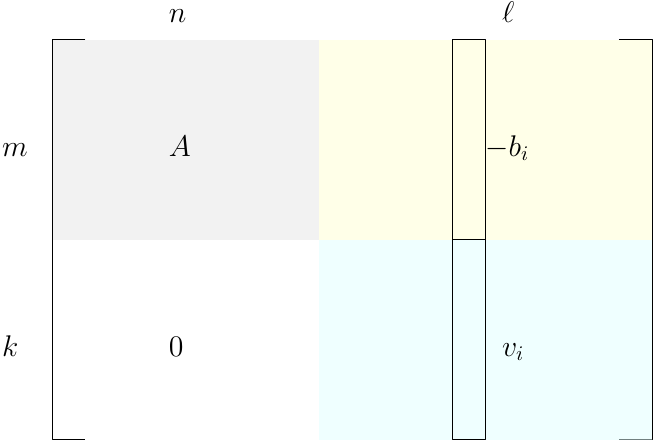}
\caption{
The matrix $A'$ in \Cref{lem:ss:zeroilp}.
}\label{fig:matrix}
\end{figure}

\begin{lemma}\label{lem:ss:zeroilp}
    \monoilp $\le_{\textsc{ppt}}$ \zeroilp.
\end{lemma}
\begin{proof}
    Consider an instance $Ax = b$ of \monoilp with $A \in \{0,1\}^{m \times n}$.
    We can assume that $A$ contains 1 in every column as otherwise such a column can be discarded.
    Let $v_1,\dots,v_\ell \in \{-1,0,1\}^k$ be the sequence of vectors from \Cref{lem:graver:contruction} with $\ell \ge n$ and $k = \Oh(\log n)$.
    Next, we can assume that $||b||_\infty \le n \le \ell$ as otherwise there can be no solution.
    We decompose $b$ into a sum $b_1 + \dots + b_\ell$ of vectors from $\{-1,0,1\}^m$, possibly using zero-vectors for padding.
    Now we construct a matrix $A' \in \{-1,0,1\}^{(m+k)\times(n+\ell)}$.
    The first $n$ columns are given by the columns of $A$ with 0 on the remaining $k$ coordinates.
    The last $\ell$ columns are of the form $(-b_i,v_i)$ for $i \in [\ell]$.
    See \Cref{fig:matrix} for an illustration.
    The new parameter is $m+k$ which is $m + \Oh(\log n)$.

    We claim that $Ax = b$ is feasible over boolean domain if and only if $A'y = 0$ admits a non-zero boolean solution $y$.
    Consider a solution $x$ to $Ax = b$.
    We define $y$ as $x$ concatenated with vector $1^\ell$.
    In each of the first $m$ rows we have $(A'y)_j = (Ax)_j - b_j = 0$.
    In the remaining $k$ rows we have first $n$ zero vectors followed by the sequence  $v_1,\dots,v_\ell$ which sums up to 0 by construction.
    Hence $A'y = 0$ while $y \ne 0$.

    Now consider the other direction and let $y \ne 0$ be a solution to $A'y = 0$.
    Let us decompose $y$ as a concatenation of $y_1 \in \{0,1\}^n$ and $y_2 \in \{0,1\}^\ell$.
    First suppose that $y_2 = 0$.
    Then $y_1 \ne 0$ and $Ay_1 = 0$ but this is impossible since  $A \in \{0,1\}^{m \times n}$ and, by assumption, every column of $A$ contains 1.
    It remains to consider the case $y_2 \ne 0$.
    By inspecting the last $k$ rows of $A'$ we infer that the non-zero indices of $y_2$ correspond to a non-empty subsequence of  $v_1,\dots,v_\ell$ summing up to 0.
    By construction, this is not possible for any proper  subsequence of $(v_1,\dots,v_\ell)$ so we must have $y_2 = 1^\ell$.
    Hence $0 = A'y = A'y_1 + A'y_2 = A'y_1 + (-b,0)$ and so $Ay_1 = b$. This concludes the proof of the reduction.
\end{proof}

We will now reduce from \zeroilp to \ilp. 
The subtlety comes from the fact that in the latter problem we accept the solution $x = 0$ while in the first we do not.
Observe that the reduction is easy when we can afford guessing a single column from a solution.
For a matrix $A \in \zz^{m \times n}$ and $i \in [n]$ we denote by $A^i \in \zz^{m \times 1}$ the $i$-th column of $A$ and by $A^{-i} \in \zz^{m \times (n-1)}$ the matrix obtained from $A$ by removal of the $i$-th column.

\begin{observation}\label{lem:ss:0sum:obs}
    An instance $Ax=0$ of \zeroilp is solvable if and only if there is $i \in [n]$
     such that the instance $A^{-i}y= -A^i$ of \ilp is solvable. 
\end{observation}

\begin{lemma}
    \zeroilp  $\le_{\textsc{nppt}}$ \ilp.
\end{lemma}
\begin{proof}
   \Cref{lem:ss:0sum:obs} enables us to solve \zeroilp in polynomial time when $\log n$ is large compared to $m$, by considering all $i \in [n]$ and solving the obtained \ilp instance.
   Hence we can again assume that $\log n \le \poly(m)$.
    In this case, \Cref{lem:ss:0sum:obs} can be interpreted as an NPPT that guesses $\poly(m)$ bits to identify the index $i \in [n]$.
\end{proof}

\begin{lemma}
    \ilp $\le_{\textsc{nppt}}$ \monoilp.
\end{lemma}
\begin{proof}
    We decompose the matrix $A$ as $A^+ - A^-$ where $A^+, A^-$ have entries from $\{0,1\}$.
    Suppose that there exists a vector $x$ satisfying $Ax = b$.
    We nondeterministically guess vectors $b^+, b^-$ that satisfy $A^+x = b^+$, $A^-x = b^-$ and we check whether $b^+ - b^- = b$; if no then the guess is rejected.
    This requires $m\log n$ nondeterministic bits.
    We create an instance of \monoilp with $2m$ constraints given as $A^+y = b^+$, $A^-y = b^-$.
    If we made a correct guess, then $y=x$ is a solution to    
    the system above.
    On the other hand, if this system admits a solution $y$ then $Ay = A^+y - A^-y = b^+ - b^- = b$ so $y$ is also a solution to the original instance.
\end{proof}

\begin{lemma}
    \sslog $\equiv_{\textsc{nppt}}$ {\sc Group-$\zz_q$ Subset Sum[$\log q$]}.
\end{lemma}
\begin{proof}
    For the reduction $(\le)$ consider $q = nt$ and leave $t$ intact.
    We can assume that each input number belongs to $[1,t)$ hence the sum of every subset belongs to $[1,q)$ and so there is no difference in performing addition in $\zz$ or $\zz_q$.

    Now we handle the reduction $(\ge)$.
    Let $S$ be the subset of numbers that sums up to $t$ modulo $q$.
    Since each item belongs to $[0,q)$ their sum in $\zz$ is bounded by $nq$; let us denote this value as $t'$.
    We nondeterministically guess $t' \in [0, nq]$ and check whether $t' = t \mod m$.
    We consider an instance $J$ of {\sc Subset Sum} over $\zz$ with the unchanged items and the target $t'$.
    We have $\log t' \le \log n + \log q$ what bounds the new parameter as well as the number of necessary
    nondeterministic bits.
    If the guess was correct then $J$ will have a solution.
    Finally, a solution to $J$ yields a solution to the original instance because $t' = t \mod q$.
\end{proof}

Using the presented lemmas, any two problems listed in \Cref{thm:main:equivalent} can be reduced to each other via NPPT.

\section{Permutation Subset Sum}
\label{sec:perm}

This section is devoted to the proof of \Cref{thm:main:permutation}.
We will use an intermediate problem involving a computational model with $\ell$ binary counters,
being a special case of bounded {\em Vector Addition System with States} (VASS)~\cite{vass}.
This can be also regarded as a counterpart of the intermediate problem used for establishing XNLP-hardness, which concerns cellular automata~\cite{BodlaenderGNS21, ElberfeldST15}.

For a sequence $\fcal = (f_1, \dots, f_n)$, $f_i \in \{O,R\}$ (optional/required),
we say that a subsequence of $[n]$ is {\em $\fcal$-restricted} if it contains all the 
indices $i$ with $f_i = R$.
We say that a sequence of vectors $v_1, \dots, v_n \in \{-1,0,1\}^\ell$ forms a {\em 0/1-run} if $v_1 + \dots + v_n = 0$ and for each $j \in [n]$ the partial sum $v_1 + \dots + v_j$ belongs to $\{0,1\}^\ell$.

\defparproblem{0-1 Counter Machine}
{Sequences $\mathcal{V} = (v_1, \dots, v_n),\, v_i \in \{-1,0,1\}^\ell$, and $\fcal = (f_1, \dots, f_n),\, f_i \in \{O,R\}$.}
{$\ell$}
{Is there a subsequence $(i_1 < i_2 < \dots < i_r)$ of $[n]$ that is $\fcal$-restricted and such that $(v_{i_1},v_{i_2},\dots,v_{i_r})$ forms a 0/1-run?}

Intuitively, a vector $v_i \in \{-1,0,1\}^\ell$ tells which of the $\ell$ counters should be increased or decreased.
We must ``execute'' all the vector $v_i$ for which $f_i = R$ plus some others so that the value of each counter is always kept within $\{0,1\}$.

We give a reduction from {\sc 3Coloring[pw]} to \vass.

\defparproblem{3Coloring}{An undirected graph $G$ and a path decomposition of $G$ of width at most {\sc pw}}{\sc pw}{Can we color $V(G)$ with 3 colors so that the endpoints of each edge are assigned different colors?}

\begin{restatable}[$\bigstar$]{lemma}{lemColoring}
\label{lem:perm:coloring}
    {\sc 3Coloring[pw]} $\le_{\textsc{ppt}}$ \vass.
\end{restatable}

In the proof, 
we assign each vertex a label from  $[\textsc{pw} + 1]$  so that the labels in each bag are distinct. 
We introduce a counter for each pair (label, color) and whenever a vertex is introduced in a bag, we make the machine increase one of the counters corresponding to its label.
For each edge $uv$ there is a bag containing both $u,v$; we then insert a suitable sequence of vectors so that running it is possible if and only if the labels of $u, v$ have active counters in different colors.
Finally, when a vertex is forgotten we deactivate the corresponding counter.

In order to encode the operations on counters as composition of permutations, we will employ the following algebraic construction.
For $q \in \nn$ consider an automorphism $\phi_1 \colon \zz_q^2 \to \zz_q^2$ given as $\phi_1((x,y)) = (y,x)$.
Clearly $\phi_1 \circ \phi_1$ is identify, so there is a homomorphism $\phi\colon\zz_2 \to Aut(\zz_q^2)$ that assigns identity to $0 \in \zz_2$ and $\phi_1$ to $1 \in \zz_2$.
We define the group $U_q$ as the outer semidirect product  $\zz_q^2 \rtimes_\phi \zz_2$ (see \Cref{sec:prelims}).
That is, the elements of $U_q$ are $\{((x,y),z) \mid x,y \in \zz_q, z \in \zz_2\}$ and the group operation $\circ$ is given as 

\begin{equation}
   ((x_1,y_1),z_1) \circ ((x_2,y_2),z_2)=
   \begin{cases}
    ((x_1 \oplus_q x_2), (y_1 \oplus_q y_2), z_1 \oplus_2 z_2)  & \text{if } z_1 = 0\\
    ((x_1 \oplus_q y_2), (y_1 \oplus_q x_2), z_1 \oplus_2 z_2)  & \text{if } z_1 = 1.
   \end{cases}
\end{equation}

The $z$-coordinate works as addition modulo 2 whereas the element $z_1$ governs whether we add $(x_2,y_2)$ or $(y_2,x_2)$ modulo $q$ on the $(x,y)$-coordinates.
The neutral element is $((0,0),0)$.
Note that $U_q$ is non-commutative.

For $((x,y),z) \in U_q$ we define its {\em norm} as $x+y$.
Consider a mapping $\Gamma \colon \{-1,0,1\} \to U_q$ given as $\Gamma(-1) = ((1,0),1)$, $\Gamma(0) = ((0,0),0)$, $\Gamma(1) = ((0,1),1)$.

\begin{lemma}\label{lem:perm:run-group-U}
    Let $b_1,\dots, b_n \in  \{-1,0,1\}$ and $q > n$.
    Then  $b_1,\dots, b_n$ forms a 0/1-run (in dimension $\ell=1$) if and only if the group product $g = \Gamma(b_1) \circ  \Gamma(b_2) \circ \dots \circ  \Gamma(b_n)$ in $U_q$ is of the form $g = ((0,n'),0)$ for some $n' \in[n]$. 
\end{lemma}
\begin{proof}
    Recall that $\Gamma(0)$ is the neutral element in $U_q$.
    Moreover, removing 0 from the sequence does not affect the property of being a 0/1-run, so we can assume that $b_i  \in \{-1,1\}$ for each $i \in [n]$.
    Note that the inequality $q > n$ is preserved by this modification.
    This inequality is only needed to ensure that the addition never overflows modulo $q$.

    Suppose now that $b_1,\dots, b_n$ is a 0/1-run.
    Then it comprises alternating 1s and -1s: $(1,-1,1,-1,\dots,1,-1)$.
    Hence the product $g = \Gamma(b_1)  \circ \dots \circ  \Gamma(b_n)$ equals $(\Gamma(1) \circ \Gamma(-1))^{n/2}$.
    We have $((0,1),1) \circ ((1,0),1) = ((0,2),0)$ and so $g = ((0,n),0)$.

    Now suppose that $b_1,\dots, b_n$ is not a 0/1-run.
    Then either $\sum_{i=1}^n b_i = 1$ or $\sum_{i=1}^j b_i \not\in \{0,1\}$ for some $j \in [n]$.
    In the first scenario $n$ is odd so $g$ has 1 on the  $z$-coordinate and so it is not in the form of $((0,n'),0)$.
    In the second scenario there are 3 cases: (a) $b_1 = -1$, (b) $(b_i,b_{i+1}) = (1,1)$ for some $i\in[n-1]$, or (c) $(b_i,b_{i+1}) = (-1,-1)$ for some $i\in[n-1]$.
    
    Case (a): $g = \Gamma(-1) \circ h = ((1,0),1) \circ h$ for some $h \in U_q$ of norm $\le n-1$. 
    Then $g$ cannot have 0 at the $x$-coordinate because $n<q$ and the addition does not overflow.
    
    Case (b): $\Gamma(1)^2 = ((0,1),1)^2 = ((0,1) \oplus_q (1,0), 1 \oplus_2 1) = ((1,1),0)$.
    For any $h_1,h_2 \in U_q$ of total norm $\le n-1$ the product $h_1 \circ ((1,1),0) \circ h_2$ cannot have 0 at the $x$-coordinate.

    Case (c): Analogous to (b) because again $\Gamma(-1)^2 = ((1,0),1)^2 =  ((1,0) \oplus_q (0,1), 1 \oplus_2 1) = ((1,1),0)$.
\end{proof}

Next, we show how to embed the group $U_q$ into a permutation group over a universe of small size.
On an intuitive level, we need to implement two features: counting modulo $q$ on both coordinates and a mechanism to swap the coordinates.
To this end, we will partition the universe into two sets corresponding to the two coordinates.
On each of them, we will use a permutation of order $q$ to implement counting without interacting with the other set.
Then we will employ a permutation being a bijection between the two sets, which will work as a switch.
See \Cref{fig:perm} for a~visualization.

\begin{figure}[t]
\centering
\includegraphics[scale=1.0]{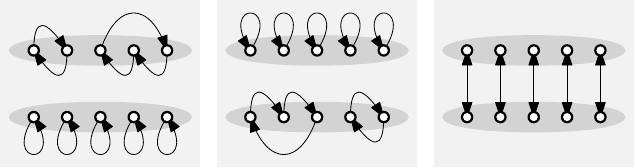}
\caption{
An illustration to \Cref{lem:perm:homo}. 
The three permutations are $\pi_1, \pi_0, \pi_z \in S_{10}$.
The first one acts as $g$ on the upper set and as identity on the lower set.
In the second one these roles are swapped whereas $\pi_z$ acts as symmetry between the two sets. 
The permutations $\widehat{\Gamma}(1), \widehat{\Gamma}(-1)$ in \Cref{lem:perm:run-group-perm} are obtained as $\pi_0 \circ \pi_z$ and $\pi_1 \circ \pi_z$.
Multiplying a sequence of permutations from $\{\widehat{\Gamma}(1), \widehat{\Gamma}(-1)\}$ yields a permutation acting as identity on the upper set if and only if the arguments are alternating~1s~and~-1s.
}\label{fig:perm}
\end{figure}

\begin{lemma}\label{lem:perm:homo}
    For every $n \in \nn$ there exist $q > n$ and $\hat r = \Oh(\log^3 n)$ for which there is a homomorphism $\chi \colon U_q \to S_{\hat r}$.
\end{lemma}
\begin{proof}
    By \Cref{lem:prelim:landau} we can find a permutation $g$ of order $q > n$ in $S_{r}$ for some $r = \Oh(\log^3 n)$.
    The subgroup of $S_r$ generated by $g$ is isomorphic to $\zz_q$.
    We will now consider the permutation group over the set $[r] \times \zz_2$, which is isomorphic to $S_{2r}$.
    Instead of writing $\chi$ explicitly, we will identify a subgroup of $S_{2r}$ isomorphic to $U_q$.
    
    Let $\pi_z$ be the permutation given as $\pi_z(i,j) = (i,1-j)$ for $(i,j) \in [r]\times \zz_2$, i.e., it switches the second coordinate.
    Let $\pi_0$ act as $g$ on $[r] \times 0$ and as identify on  $[r] \times 1$.
    Analogously, let $\pi_1$ act as $g$ on $[r] \times 1$ and as identify on  $[r] \times 0$.
    Let $N$ be the subgroup of $S_{2r}$ generated by $\pi_0$ and $\pi_1$; it is isomorphic to $\zz_q^2$ and each element of $N$ is of the form $(\pi_0^x, \pi_1^y)$ for some $x,y \in \zz_q$.
    Next, let $H$ be the subgroup generated by $\pi_z$; it is isomorphic to $\zz_2$.
    Now consider a homomorphism $\phi \colon H \to Aut(N)$ given as conjugation  $\phi_\pi(g) = \pi \circ g \circ \pi^{-1}$.
    In this special case, the semidirect product $N \rtimes_\phi H$ is isomorphic to the subgroup of $S_{2r}$ generated by the elements of $N$ and $H$ (\Cref{lem:prelim:semidirect}).
    On the other hand, $\phi_{\pi_z}$ maps $(\pi_0^x, \pi_1^y) \in N$ into $(\pi_0^y, \pi_1^x)$ so this is exactly the same construction as used when defining $U_q$.
    We infer that $U_q$ is isomorphic to a subgroup of $S_{2r}$ and the corresponding homomorphism is given by the mapping of the generators: $\chi((0,1),0) = \pi_0$, $\chi((1,0),0) = \pi_1$, $\chi((0,0),1) = \pi_z$.
    
\end{proof}
Armed with such a homomorphism, we translate \Cref{lem:perm:run-group-U} to the language of permutations.

\begin{lemma}\label{lem:perm:run-group-perm}
    For every $n \in \nn$ there exists $r = \Oh(\log^3 n)$, a permutation $\pi \in S_r$ of order greater than $n$, and a mapping $\widehat\Gamma \colon \{-1,0,1\} \to S_r$ so that the following holds.
    A sequence $b_1,\dots, b_n \in  \{-1,0,1\}$ is a 0/1-run if and only if the product $\widehat\Gamma(b_1) \circ  \widehat\Gamma(b_2) \circ \dots \circ  \widehat\Gamma(b_n)$ is of the form $\pi^{n'}$ for some $n' \in [n]$.
\end{lemma}
\begin{proof}
    Let $\chi \colon U_q \to S_r$ be the homomorphism from \Cref{lem:perm:homo} for $q > n$ and $r = \Oh(\log^3 n)$.
    We define $\widehat\Gamma \colon \{-1,0,1\} \to S_r$
    as $\widehat\Gamma(i) = \chi(\Gamma(i))$ using the mapping $\Gamma$ from \Cref{lem:perm:run-group-U}.
    Since $\chi$ is a homomorphism, the condition $\Gamma(b_1) \circ \dots \circ \Gamma(b_n) = ((0,n'),0)$ is equivalent to $\widehat\Gamma(b_1) \circ \dots \circ \widehat\Gamma(b_n) = \chi(((0,n'),0))$.
    We have $g = \chi(((0,n'),0))$ for some $n' \in [n]$ if and only if $g = \pi^{n'}$ 
    for $\pi = \chi((0,1),0)$.
    The order of $\pi$ is $q > n$, as requested.
\end{proof}

For a sequence of vectors from $\{0,1\}^\ell$ we can use a Cartesian product of $\ell$ permutation groups $S_r$ to check the property of being a 0/1-run by inspecting the product of permutations from $S_{\ell r}$. 
This enables us to encode the problem with binary counters as {\sc Group-$S_k$ Subset Sum[$k$]}.
We remark that we need nondeterminism to guess the target permutation.
This boils down to guessing
the number $n'$ from \Cref{lem:perm:run-group-perm} for each of $\ell$ coordinates.

Finally, \Cref{thm:main:permutation} follows by combining \Cref{lem:perm:coloring} with \Cref{thm:perm:final}.

\begin{restatable}[$\bigstar$]{lemma}{lemPermFinal}
\label{thm:perm:final}
    \vass $\le_{\textsc{nppt}}$ {\sc Group-$S_k$ Subset Sum[$k$]}.
\end{restatable}

\section{Conclusion}

We have introduced the nondeterministic polynomial parameter transformation (NPPT) and used this concept to shed some light on the unresolved questions about short certificates for FPT problems.
We believe that our work will give an impetus for further systematic study of certification complexity in various contexts.

The main question remains to decipher certification complexity of \sslog.
Even though {\sc Subset Sum} enjoys a seemingly simple structure, some former breakthroughs required advanced techniques such as additive combinatorics~\cite{AbboudBHS22, chen2024approximating} or number theory~\cite{Kane10}.
\Cref{thm:main:equivalent} makes it now possible to analyze \sslog through the geometric lens using concepts such as lattice cones~\cite{EisenbrandS06} or Graver bases~\cite{berndt2024new, berstein2009graver, kudo2013lower}.

Drucker et al.~\cite{DruckerNS16} suggested also to study {\sc $k$-Disjoint Paths} and {\sc $K$-Cycle} in their regime of polynomial witness compression.
Recall that the difference between that model and ours is that they ask for a \emph{randomized} algorithm that \emph{outputs} a solution.
Observe that a polynomial certificate (or witness compression) for  {\sc $k$-Disjoint Paths} would entail an algorithm with running time $2^{k^{\Oh(1)}}n^{\Oh(1)}$ which seems currently out of reach~\cite{kawarabayashi2012disjoint}.
What about a certificate of size $(k + \log n)^{\Oh(1)}$?

Interestingly, {\sc Planar $k$-Disjoint Paths} does admit a polynomial certificate: if $k^2 \le \log n$ one can execute the known $2^{\Oh(k^2)}n^{\Oh(1)}$-time algorithm~\cite{LokshtanovMP0Z20} and otherwise one can guess the homology class of a solution (out of $n^{\Oh(k)} \le 2^{\Oh(k^3)}$) and then solve the problem in polynomial time~\cite{schrijver1994finding}.
Another interesting question is whether {\sc $k$-Disjoint Paths} admits a certificate of size $(k + \log n)^{\Oh(1)}$ on acyclic digraphs.
Note that we need to incorporate $(\log n)$ in the certificate size because the problem is W[1]-hard when parameterized by $k$~\cite{Slivkins10}.
The problem admits an $n^{\Oh(k)}$-time algorithm based on dynamic programming~\cite{fortune1980directed}.

For {\sc $K$-Cycle} we cannot expect to rule out a polynomial certificate via a PPT from \andsat because the problem admits a polynomial compression~\cite{Wahlstrom13Abusing}, a property unlikely to hold for \andsat~\cite{fomin2019kernelization}.
Is it possible to establish the 
certification hardness by NPPT (which does not preserve polynomial compression) or
would such a reduction also lead to unexpected consequences?


A different question related to bounded nondeterminism is whether 
one can rule out a logspace algorithm for 
directed reachability (which is NL-complete) using only polylog($n$) nondeterministic bits.
Observe that relying on the analog of the assumption \conp for NL would be pointless because NL $=$ coNL by Immerman-Szelepcsényi Theorem.
This direction bears some resemblance to the question whether directed reachability can be solved in polynomial time and polylogarithmic space, i.e., whether NL $\sub$ SC~\cite{aaronson2005complexity}.

\bibliographystyle{plainurl}
\bibliography{main}

\appendix

\section{Scheduling}

A {\em job} is represented by a triple $(p_i,w_i,d_i)$ of positive integers (processing time, weight, due date).
For a sequence of $n$ jobs a {\em schedule} is a permutation $\rho \colon [n] \to [n]$.
The {\em completion time} $C_i$ of the $i$-th job in a schedule $\rho$ equals $\sum_{j \in [n],\, \rho(j) \le \rho(i)} p_j$.
A job is called {\em tardy} is $C_i > d_i$. 
A job is being {\em processed} at time $x$ if $x \in (C_i - p_i, C_i]$.

\defparproblem{Scheduling Weighted Tardy Jobs}{A sequence of jobs  $(p_1,w_1,d_1), (p_2,w_2,d_2),\dots,(w_n,p_n,d_n)$, integer $w$}{$\log (d_{max}+w_{max})$}{Is there a schedule in which the total weight of tardy jobs is at most $w$?}

We show that {\sc Scheduling Weighted Tardy Jobs} does not admit a polynomial certificate as long as \conp and \Cref{conj:certificate} holds.
The proof utilizes a construction that appeared in a conditional lower bound on pseudo-polynomial running time for {\sc Scheduling Weighted Tardy Jobs}~\cite[Lemma 3.4]{AbboudBHS20}.
On an intuitive level, we will pack multiple instances of {\sc Subset Sum} into a single instance of {\sc Scheduling Weighted Tardy Jobs} and adjust the weights to enforce that all the {\sc Subset Sum} instances must be solvable if a good schedule exists.

\thmMainScheduling*
\begin{proof}
    Consider an instance \andsat given by a sequence of formulas $\Phi_1, \dots, \Phi_n$ of {\sc 3SAT}, each on at most $k$ variables.
    If $2^k \le n$ then the running time $2^k\cdot\poly(k)\cdot n$ becomes polynomial in $n$ so we can assume that $n \le 2^k$.
    We transform each instance $\Phi_j$ into an equivalent instance $(S_j,t_j)$ of {\sc Subset Sum} following the standard NP-hardness proof for the latter~\cite{jansen2016bounding}.
    The number of items in $S_j$, as well as logarithm of $t_j$, becomes linear in the number of variables plus the number of clauses  in $\Phi_j$, which is $\Oh(k^3)$.
    
    Now we create an instance $I$ of {\sc Scheduling Weighted Tardy Jobs} with $\sum_{i=1}^n |S_i|$ items.
    Let $s_j = t_1 + \dots + t_j$ for $j \in [n]$ and $s_0 = 0$.
    For each $j \in [n]$ and $x \in S_j$ we create a job $u$ with due date $d_u = s_j$, processing time $p_u = x$, and weight $w_u = p_u \cdot (n+1-j)$; we say that such job $u$ comes from the $j$-th set.
    We claim that $I$ admits a schedule with total weight of non-tardy jobs at least $\sum_{i=1}^n t_i\cdot(n+1-i)$ if and only if all $n$ instances  $(S_j,t_j)$ of {\sc Subset Sum} are solvable.
    
    The implication $(\Leftarrow)$ is easy.
    For $j\in [n]$ let $S'_j \sub S_j$ be the subset summing up to $t_j$.
    We schedule first the jobs corresponding to the elements of $S'_1$, then the ones corresponding to $S'_2$, and so on.
    The tardy jobs are scheduled in the end in an arbitrary order.
    Then the jobs coming from the $j$-th set are being processed within the interval $(s_{j-1}, s_j]$ of length $t_j$ so they all meet their deadlines.
    The weight of the non-tardy jobs from the $j$-th set equals $\sum_{x \in S'_j} x \cdot (n+1-j) = t_j \cdot  (n+1-j)$ and the total weight is as promised.

    Now we prove the implication $(\Rightarrow)$.
    Consider some schedule of jobs in the instance $I$ with total weight of the non-tardy jobs at least $\sum_{i=1}^n t_i\cdot(n+1-i)$.
    We define two sequences $(a_i), (b_i)$ of length $s_n$.
    We set $a_i = (n+1-j)$ where $j \in [n]$ is the unique index satisfying $i \in (s_{j-1}, s_j]$.
    If there is a job $u$ being processed at time $i$ we set $b_i = w_u/p_u$, otherwise we set $b_i = 0$.
    We claim that for each $i \in [s_n]$ it holds that $b_i \le a_i$.
    Indeed, when $i \in (s_{j-1}, s_j]$ then any non-tardy job $u$ processed at time $i$ must have its due date at $s_j$ or later.
    Hence $u$ comes from the $j'$-th set, where $j' \ge j$, so $b_i = w_u/p_u = n+1-j' \le n+1-j = a_i$.
    Now observe that $\sum_{i=1}^t a_i =  \sum_{j=1}^n t_j\cdot(n+1-j)$ whereas $\sum_{i=1}^t b_i$ equals the total weight of the non-tardy jobs in the considered schedule.
    By assumption, we must have $\sum_{i=1}^t a_i = \sum_{i=1}^t b_i$ and so $a_i = b_i$ for all $i \in [t]$.
    Therefore, for each $j \in [n]$ and  $i \in (s_{j-1}, s_j]$ we have $b_i = (n+1-j)$ so the total processing time of the non-tardy jobs coming from the $j$-set equals $t_j$.
    Hence each instance $(S_j,t_j)$ of {\sc Subset Sum} admits a solution.

    It remains to check that the parameter $\log (d_{max} + w_{max})$ is polynomially bounded in~$k$.
    The maximum due date $d_{max}$ equals $s_n \le n \cdot 2^{\Oh(k^3)}$.
    The maximum weight $w_{max}$ is bounded by $n$ times the maximum size of an element in $\bigcup_{i=1}^n S_i$ which is $2^{\Oh(k^3)}$.
    By the initial assumption $n \le 2^k$ so  $\log (d_{max} + w_{max}) \le \poly(k)$.
    This concludes the proof.   
\end{proof}

\section{Permutation Subset Sum: Missing proofs}

\lemColoring*
\begin{proof}
    Let $k = \textsc{pw} + 1$.
    By \Cref{prelim:pathwidth:nice} we can assume that the given path decomposition of the graph $G$ is nice.
    We can turn it into a sequence $C$ of commands of the form $\intro(v), \forget(v), \edge(u,v)$ such that (i) each vertex $v$ is introduced once and forgotten once afterwards, and (ii) for each edge $uv \in E(G)$ the command $\edge(u,v)$ appears at some point when $u,v$ are present (i.e., after introduction and before being forgotten). 
    Next, we assign each vertex a label from $[k]$ in such a way that no two vertices of the same label share a bag.
    To this end, we scan the sequence $C$, store the current bag, and assign the colors greedily.
    Whenever we see a command $\intro(v)$ then the current bag must contain less than $k$ vertices so there is some label not being used in the current bag.
    We then assign this label to $v$ and continue.

    We will now assume that the arguments of the commands  $\intro, \forget, \edge$ refer to labels, e.g., $\intro(x)$ means that a vertex with label $x$ appears in the current bag and $\edge(x,y)$ means that the vertices labeled $x,y$ in the current bag are connected by an edge.

    Let $D \sub [3]^2$ be the set of all 6 ordered pairs different colors.
    For each label $x \in [k]$ and each color $c \in [3]$ we create a counter called $x_c$.
    The intended meaning if $x_c = 1$ is that the vertex labeled $x$ in the current bag is assigned color $c$.
    We also create special counters: one named $S$ 
    and 6 counters $Z_{c,d}$ for each $(c,d) \in D$.
    Hence the total number of counters is $\ell = 3k + 7$.

    We translate the sequence $C$ into sequences $\mathcal{V}$, $\fcal$ forming the instance of \vass.
    Instead of specifying $\fcal$ directly, we will indicate which vectors in $\mathcal{V}$ are {\em optional} and which are {\em required}.
    To concisely describe a vector in which, e.g., counter $y_1$ is being increased, counter $y_2$ is being decreased, and the remaining ones stay intact, we write
    $[y_1\uparrow,\, y_2\downarrow]$. 
    We scan the sequence $C$ and for each command we output a {\em block} of vectors, according to the following instructions.

    \begin{itemize}
        \item [$\intro(x)$:] We insert optional vectors $[x_1\uparrow,\, S\uparrow]$, $[x_2\uparrow,\, S\uparrow]$, $[x_3\uparrow,\, S\uparrow]$, followed by a required vector $[S\downarrow]$.
        
        Since the counter $S$ cannot exceed 1, at most one of the first three vectors can be used.
        In the end we are required to decrease $S$ so if we began with $S$ set to 0 then exactly one of $x_1, x_2, x_3$ must be set to 1.

        \item [$\forget(x)$:] We insert optional vectors $[x_1\downarrow,\, S\uparrow]$, $[x_2\downarrow,\, S\uparrow]$, $[x_3\downarrow,\, S\uparrow]$, followed by a required vector $[S\downarrow]$.

        Similarly as before, we can decrease at most one of the counters $x_1, x_2, x_3$ and we will check that exactly one this events must happen.
        
        \item [$\edge(x,y)$:] 
        We insert 6 optional vectors $[x_c\downarrow,\, y_d\downarrow,\, Z_{c,d}\uparrow,\, S\uparrow]$, one for each $(c,d) \in D$.
        This is followed by required vectors $[S\downarrow]$, $[S\uparrow]$.
        Then we insert 6 optional vectors $[x_c\uparrow,\, y_d\uparrow,\, Z_{c,d}\downarrow,\, S\downarrow]$, for each $(c,d) \in D$, this time followed by $[S\uparrow]$, $[S\downarrow]$.

        Suppose that initially $S$ is set to 0.
        The required vectors manipulating $S$ enforce that exactly one of the first 6 vectors and exactly one of the last 6 vectors are used.
        This verifies that the vertices labeled $x,y$ are colored with different colors. 
    \end{itemize}

\begin{claim}
If the graph is 3-colorable then there exists a subsequence of $\mathcal{V}$ whose indices are $\fcal$-restricted and which forms a 0/1-run.
\end{claim}
\begin{claimproof}
    We will maintain the following invariant: between the blocks (1) all the special counters are set to 0, (2) for each vertex in the current bag with label $x$ and color $c$ the counter $x_c$ is set to 1.
    
    When a vertex with label $x$ and color $c$ is introduced we choose the optional vector $[x_c\uparrow,\, S\uparrow]$ and subsequently $S$ gets decreased.
    When a vertex with label $x$ and color $c$ is forgotten we know that $x_c$ is currently set to 1 so we can choose the optional vector $[x_c\downarrow,\, S\uparrow]$ and then again $S$ gets decreased.
    When the command $\edge(x,y)$ is processed we know that there is a pair $(c,d) \in D$ so the counters $x_c, y_d$ are set to 1.
    Therefore we can execute the sequence $[x_c\downarrow,\, y_d\downarrow,\, Z_{c,d}\uparrow,\, S\uparrow]$, $[S\downarrow]$, $[S\uparrow]$, $[x_c\uparrow,\, y_d\uparrow,\, Z_{c,d}\downarrow,\, S\downarrow]$, $[S\uparrow]$, $[S\downarrow]$, maintaining the invariant in the end.
    Finally, when a vertex labeled $x$ with color $c$ is being forgotten then the counter $x_c$ is decreased so in the end all the counters are 0.
\end{claimproof}

\begin{claim}
If there exists a subsequence of $\mathcal{V}$ whose indices are $\fcal$-restricted and which forms a 0/1-run, then the graph is 3-colorable.
\end{claim}
\begin{claimproof}
We will prove two invariants about the states of counters between the blocks: (1) all the special counters are set to 0, (2) for each $x \in [k]$ if the vertex labeled $x$ is present in the current bag then there is exactly one $c \in [3]$ so that $x_c$ is set to 1, and if there is no vertex labeled $x$ in the current bag then all $x_1,x_2,x_3$ are set to 0.

We first prove (1).
Observe that each block ends with required $[S\downarrow]$ so it remains to analyze the counters $Z_{c,d}$.
They can be activated only in a block corresponding to command $\edge(x,y)$.
The required actions on $S$ enforce that that exactly one of the first 6 optional vectors and exactly one of the last 6 optional vectors are used.
Since each vector in second group decreases some counter $Z_{c,d}$ and the first group can increase only one, these two vectors must process the same pair $(c,d) \in D$.
Hence in the end all the special counters are again deactivated.

We move on to invariant (2).
First consider a command $\intro(x)$. Due to the initial preprocessing, no other vertex labeled $x$ can be present in the current bag so the  counters  $x_1,x_2,x_3$ are inactive. 
We must activate exactly one of the counters $x_1,x_2,x_3$ so the invariant is preserved. An analogous argument applies to $\forget(x)$.
Now consider the block corresponding to command $\edge(x,y)$.
By the argument from the the analysis of invariant (1), the two optional vectors used in this block process the same pair $(c,d) \in D$.
So the state of  all the counters after processing this block is the same as directly before it. 

We define the 3-coloring of the graph as follows.
When a vertex $v$ with label $x$ is being introduced, we know that one counter $x_c$ for some $c \in [3]$ gets increased.
We assign the color $c$ to $v$.
To check that this is a correct coloring, consider an edge $uv$.
Let $x,y$ be the labels of $u,v$. 
There is a command $\edge(x,y)$ in $C$ being executed when both $u,v$ are present in the current bag.
By the analysis above, we must take one of the vectors $[x_c\downarrow,\, y_d\downarrow,\, Z_{c,d}\uparrow,\, S\uparrow]$ at the beginning of the corresponding block, which implies that the counters $x_c$ and $y_d$ were active at the beginning of the block.
But $(c,d) \in D$ so $c\ne d$ and we infer that $u,v$ must have been assigned different colors.
\end{claimproof}
This concludes the correctness proof of the reduction.  
\end{proof}

\lemPermFinal*
\begin{proof}
    Consider an instance of \vass given by the sequences $\mathcal{V} = (v_1, \dots, v_n)$, $v_i \in \{-1,0,1\}^\ell$, and $\fcal = (f_1, \dots, f_n)$, $f_i \in \{O,R\}$.
    The problem can be easily solved in time $\Oh(2^\ell\cdot n)$ so we can assume from now on that $\log n \le \ell$.
    Let $r = \Oh(\log^3 n)$, $\pi \in S_r$, and $\widehat\Gamma \colon \{-1,0,1\} \to S_r$ be as in \Cref{lem:perm:run-group-perm}.

    For each $i \in [n]$ we construct an $(\ell+1)$-tuple of permutations $(p^1_i, p^2_i, \dots, p^{\ell+1}_i)$ from $S_r$ as follows.
    For each $j \in [\ell]$ we set $p^j_i = \widehat\Gamma(v_i^j)$ where $v_i = (v_i^1,\dots, v_i^\ell)$.
    We set $p^{\ell+1}_i = \pi$ if $f_i = R$ (i.e., the $i$-th vector is required) and otherwise we set $p^{\ell+1}_i$ to the identify permutation.
    Let $f_C$ denote the number of indices with $f_i = R$.
    Next, let $I = (i_1 < \dots < i_m)$ denote a subsequence of $[n]$.
    We claim that the following conditions are equivalent.
    \begin{enumerate}
        \item The subsequence $I$ is $\fcal$-restricted.
        \item The subsequence of $p^{\ell+1}_1,\dots, p^{\ell+1}_n$ given by indices $I$ yields product $\pi^{f_C}$.
    \end{enumerate}
    To see this, observe that $\pi$ has order larger than $n$ (as guaranteed by \Cref{lem:perm:run-group-perm}) so the product  $\pi^{f_C}$ is obtained exactly when $I$ contains all $f_C$ indices for which $f_i = R$.

    We move on to the next equivalence.
    \begin{enumerate}
        \item The subsequence of $v_1, \dots, v_n$ given by indices $I$ forms a 0/1-run.
        \item For each $j \in [\ell]$ the subsequence of $p^j_1,\dots, p^j_n$ given by indices $I$ yields product $\pi^{n_i}$ for some $n_i \in [n]$.
    \end{enumerate}
    A sequence of $\ell$-dimensional vectors forms a 0/1-run if and only if each single-dimensional sequence (corresponding to one of $\ell$ coordinates) forms a 0/1-run.
    Fix $j \in [\ell]$.
    By \Cref{lem:perm:run-group-perm} the subsequence of $v^j_1, \dots, v^j_n$ given by indices $I$ forms a 0/1-run if only only if multiplying their images under the mapping $\widehat\Gamma$ yields a product of the form $\pi^{n_i}$ for some $n_i \in [n]$.
    This justifies the second equivalence.

    To create an instance of {\sc Group-$S_k$ Subset Sum[$k$]} we take $k = (\ell +1)\cdot r$ so we can simulate multiplication in $S_r^{\ell+1}$ by dividing $[k]$ into $(\ell+1)$ subsets of size $r$.
    We have $k = \Oh(\ell^4)$ so indeed the new parameter is polynomial in $\ell$.
    
    For each $i \in [n]$ we transform the tuple $(p^1_i, p^2_i, \dots, p^{\ell+1}_i)$ into $\widehat{p}_i \in S_k$ using the aforementioned natural homomorphism.
    We nondeterministically guess the numbers $n_1,\dots, n_\ell \in [n]$ 
    and set the target permutation $p_T$ as the image of $(\pi^{n_1}, \pi^{n_2}, \dots, \pi^{n_\ell}, \pi^{f_C})$ under the natural homomorphism.
    This requires guessing $\ell \cdot \log n \le \ell^2$ bits.   
    By the equivalences above, the instance of \vass is solvable if and only there is a tuple $(n_1,\dots, n_\ell)$ for which the constructed instance $((\widehat{p}_1,\dots,\widehat{p}_n), p_T)$ of {\sc Group-$S_k$ Subset Sum[$k$]} is solvable.
    This concludes the proof.
\end{proof}

\section{Remaining proofs}
\label{sec:remaining}

First, we show that {\sc Group-$G$ Subset Sum} becomes easy for the group family $\zz_k^k$.
To this end, we need the following result from the 0-sum theory (see~\cite{zero-sum-survey} for a survey).

\begin{theorem}[\cite{MESHULAM1990197}]
\label{thm:remain:0sum}
    Let $G$ be a finite commutative group, $m$ be the maximal order of an element in $G$, and $s$ satisfy $s>m(1 + \log(|G|\,/\,m))$.
    Then any sequence $a_1 \dots, a_s$ of elements in $G$ has a non-empty subsequence that sums to zero. 
\end{theorem}

\thmZkk*
\begin{proof}
    We apply \Cref{thm:remain:0sum} to the group $G = \zz_k^k$.
    The maximal order $m$ in $G$ is $k$ and so we can set $s = k^2\log k$.
    Consider an instance of {\sc Group-$\zz_k^k$ Subset Sum[$k$]} and a solution $a_1 + \dots + a_\ell = t$ that minimizes $\ell$.
    If $\ell \ge  s$ then there exists a non-empty subsequence of $a_1, \dots, a_\ell$ that sums to 0.
    Removing this subsequence from  $a_1, \dots, a_\ell$ does not modify the sum so we obtain a shorter solution, which yields a contradiction.
    Hence $\ell < s$ and so we can guess the solution using $\ell\cdot (k\log k) = k^3\log^2 k$ bits.
\end{proof}

Next, we prove that \sslog becomes easy when we drop the restriction that each element can be used by a solution only once.
Recall that in {\sc Unbounded Subset Sum} we ask for a multiset of integers from $\{p_1, p_2,\dots,p_n\}$ that sums up to $t$.
In fact, in this variant there always exists a solution with small support.
This is a special case of the integer version of Carathéodory's theorem~\cite{EisenbrandS06}.

\begin{lemma}\label{lem:remain:unbounded}
    {\sc Unbounded Subset Sum[$\log t$]} admits a polynomial certificate.
\end{lemma}
\begin{proof}
    A solution can be represented as a multiset of input integers.
    Consider a solution $S$ that minimizes the number $\ell$ of distinct integers.
    Suppose for the sake of contradiction that $\ell > \log t$ and let $S_D$ denote the set of distinct elements in $S$.
    There are $2^\ell > t$ subsets of $S_D$ and each of them has a sum in the interval $[0,t]$.
    Consequently, there are two subsets $S_1, S_2 \sub S_D$ that give the same sum.
    The same holds for $S_1 \sm (S_1 \cap S_2)$ and $S_2 \sm (S_1 \cap S_2)$ so we can assume that $S_1,S_2$ are disjoint.
    Let $x \in S_1 \cup S_2$ be an element with the least multiplicity in $S$; let $m$ denote this multiplicity and assume w.l.o.g. that $x \in S_1$.
    We construct a new solution $S'$ from $S$: we decrease the multiplicity of each element from $S_1$ by $m$ and increase multiplicity of each element from $S_2$ by $m$.
    By the choice of $m$, the multiplicity of $x$ drops to 0 and we never decrease the multiplicity of any element below 0.
    Therefore, $S'$ is also a valid solution with a lower number of distinct elements; a contradiction.

    We have shown that $\ell \le \log t$.
    Moreover, the multiplicity of each element in $S$ cannot exceed~$t$.
    We can guess this solution by guessing the set $S_D$ and the numbers from $[t]$ representing their multiplicities, using $\log^2 t$ nondeterministic bits.
\end{proof}

Finally, we justify why we should not expect a PPT from {\sc 3-Coloring[pw]} to \andsat.
In fact, our argument works already for parameterization by treedepth.
We consider the standard {\sc CNF-SAT} problem with unbounded arity parameterized by the number of variables
~$n$.
It is known that {\sc CNF-SAT[$n$]}  is MK[2]-hard~\cite{HermelinKSWW15} and as such it is considered unlikely to admit a polynomial Turing kernelization.
On the other hand, \andsat admits a trivial polynomial Turing kernelization
and so it suffices to show a PPT from {\sc CNF-SAT[$n$]} to {\sc 3-Coloring[pw]}.
An analogous hardness has been observed for {\sc Indepedent Set} parameterized by treewidth~\cite{HolsKP20}.

\begin{lemma}\label{lem:remain:3color-mk2}
    {\sc CNF-SAT[$n$]} $\le_{\textsc{ppt}}$ {\sc 3-Coloring[pw]}. Consequently, {\sc 3-Coloring[pw]} is MK[2]-hard.
\end{lemma}
\begin{proof}
    We adapt the known NP-hardness proof of {\sc 3-Coloring} and we refer to the 3 colors as T (true), F (false), B (blocked).
    Let $g_B, g_F$ be two special vertices, connected by an edge, and let B, F refer to the colors used by $g_B, g_F$ respectively.
    
    We will utilize the 2-OR-gadget that, for given vertices $u,v$, uses a fresh vertex $g_{uv}$ (the output), so that (i) $u,v$ must be colored with T or F,
    (ii) the color of $g_{uv}$ must be F if both $u,v$ have color F,
    (iii) the vertex $g_{uv}$ can be colored with T if  one of $u,v$ has color T.
    First, connect both $u,v$ to $g_B$ to ensure (i).
    Next, create auxiliary vertices $u',v'$ and edges $uu'$, $vv'$, $u'v'$.
    The vertex $g_{uv}$ is connected to $u',v',g_B$.
    Suppose that both $u,v$ have color F.
    Then the colors used by $u',v'$ must be $\{T,B\}$ and the only color left for $g_{uv}$ is F.
    Next, if one of $u,v$ has color T then  $u',v'$ can be colored with $\{F,B\}$ and $g_{uv}$ can use T.

    We can use the output of the 2-OR-gadget as an input of the next gadget and so for each $d \in \nn$ we can construct a $d$-OR-gadget that takes $d$ vertices, uses $\Oh(d)$ auxiliary vertices, and its output implements the OR-function of the input colors.

    We are now in position to give the reduction from {\sc CNF-SAT[$n$]}.
    For each variable $x_i$ we create vertices $x_i^Y$, $x_i^N$, corresponding to literals $x_i, \neg x_i$, and connect them to the special vertex $g_B$.
    For each clause $\phi$ with arity $d$ we create a $d$-OR-gadget over the vertices corresponding to its literals, and connect its output to the special vertex $g_F$.
    First, we analyze the pathwidth of such a graph $G$ by constructing a path decomposition.
    We put all the literal and special vertices ($2n+2$ in total) in all the bags.
    Then each bag corresponds to some clause and contains the respective OR-gadget on $\Oh(n)$ vertices.
    Consequently, we obtain a path decomposition of $G$ of width~$\Oh(n)$.

    Suppose now that the given instance from {\sc CNF-SAT[$n$]} admits a satisfying assignment.
    If the variable $x_i$ is set to True, we color $x_i^Y$ with T and $x_i^N$ with F.
    Otherwise we use the opposite coloring.
    By the property (iii) of the OR-gadget, there is a coloring that assigns T to each output of the gadget and so we obtain a proper 3-coloring.

    In the other direction, consider a proper 3-coloring of $G$.
    Let B,F be the colors used by $g_B,g_F$.
    Each pair $(x_i^Y,x_i^N)$ is colored as T/F or F/T; we treat $x_i$ as set to True when $x_i^Y$ uses color T.
    The output of each OR-gadget must utilize the color unused by $g_B,g_F$, that is, T.
    By the property (ii) of the OR-gadget, one of the inputs also must be colored with T.
    Hence the described truth-assignment satisfies every clause. 
\end{proof}

\end{document}